\documentclass[12pt, letter]{article}

\usepackage{times}
\usepackage{bm}
\usepackage{natbib}
\usepackage{verbatim}
\usepackage{graphicx}
\usepackage{amsmath}
\usepackage{subfigure}
\usepackage{color}
\usepackage{ulem}
\usepackage{amsthm} 
\usepackage{rotating}

\def\T{{ \mathrm{\scriptscriptstyle T} }}
\def\d{ {\text{d}} }

\newtheorem{theorem}{Theorem}

\newtheorem{assumption}{Assumption}
\newtheorem{lemma}{Lemma}

\title{Nonparametric Bayes modeling of count processes}
\author{Antonio Canale\thanks{Department of Economics and Statistics, University  of Turin and Collegio Carlo Alberto, Italy  \mbox{(\emph{antonio.canale@unito.it}}}  $\,$ \& David B. Dunson\thanks{
			Dept. Statistical Science, Duke University, Durham, NC 27708, USA  \mbox{(\emph{dunson@stat.duke.edu})}}}
\date{}

\begin{document}

\maketitle

\abstract{
Data on count processes arise in a variety of applications, including longitudinal, spatial and imaging studies measuring count responses.  The literature on statistical models for dependent count data is dominated by models built from hierarchical Poisson components.  The Poisson assumption is not warranted in many applications, and hierarchical Poisson models make restrictive assumptions about over-dispersion in marginal distributions.  This article proposes a class of nonparametric Bayes count  process models, which are constructed through rounding real-valued underlying processes.  The proposed class of models accommodates applications in which one observes separate count-valued functional data for each subject under study.  Theoretical results on large support and posterior consistency are established, and computational algorithms are developed using Markov chain Monte Carlo.  The methods are evaluated via simulation studies and illustrated through application to longitudinal tumor counts and asthma inhaler usage.}

{\center \textbf{Keywords: }}
Count functional data; Generalized linear mixed model; Hierarchical model; Longitudinal data; Poisson; Splines; Stochastic process.

\section{Introduction}

A stochastic process $y = \{ y(s), s \in \mathcal{S}\}$ is a collection of random variables indexed by $s \in \mathcal{S}$, with the domain $\mathcal{S}$ commonly corresponding to a set of times or spatial locations and $y(s)$ to a random variable observed at a specific time or location $s$.  There is a rich frequentist and Bayesian literature on stochastic processes, with common choices including Gaussian processes and L\'evy processes, such as the Poisson, Wiener, beta or gamma process.  Gaussian processes provide a convenient and well studied choice when $y: \mathcal{S} \to \Re$ is a continuous function.  In the Bayesian literature, there have been substantial computational and theoretical advances for Gaussian process models in recent years.  For example, \citet{bane:etal:2008} and \citet{murr:adam:2010} develop improved methods for posterior computation, while \citet{ghos:roy:2006} and  \citet{vdva:vzel:2009} study asymptotic properties including posterior consistency and rates of convergence.  The Gaussian process is appealing in providing a prior that can be specified to generate functions that are within an arbitrarily small neighborhood of any continuous function with positive probability \citep{ghos:roy:2006}, while also being computationally convenient. 

Our interest focuses on the case in which $y: \mathcal{S} \to \mathcal{N} = \{0,\ldots,\infty\}$, so that $y$ is a count-valued stochastic process over the domain $\mathcal{S}$.  There are many applications of such processes including epidemiology studies monitoring a count biomarker or health response over time for patients and ecology studies recording the number of birds of a given species observed at different locations.  Although there is a rich literature on count stochastic process models for longitudinal and spatial data, most models rely on Poisson hierarchical specifications.  For example, \citet{frhw:wagn:2006} consider $y(s) \sim \mbox{Poisson}\{ \lambda(s) \}$ with the Poisson mean $\lambda(s)$ varying over time according to a latent process.  \citet{rue:mart:chop:2009} recently developed an integrated nested Laplace approximation to the posterior for a broad class of latent Gaussian structured additive regression models. The observed variables are assumed to belong to an exponential family (Poisson for counts), with the means given an additive model having Gaussian and Gaussian process priors on the unknown components.  

Although such models have a flexible mean structure, the Poisson assumption is restrictive in limiting the variance to be equal to the mean, with over-dispersion introduced in marginalizing out the latent processes.  This leads to a pitfall in which the dependence structure is confounded with the degree of overdispersion in the marginals in that both are induced through the latent process. Such modeling frameworks cannot accommodate correlated count data that are under-dispersed, and substantial bias can potentially result for non-Poisson over-dispersed data. 
Relying on a hierarchical Faddy model  \citep{fadd:1997}, \citet{grun:bruc:2011} developed methods that generalize the Poisson to accommodate under- and over-dispersed longitudinal counts.  The Faddy distribution allows the current rate of occurrence to depend on the number of events in a previous interval, and when a dispersion parameter is less than zero the rate decreases with each new event causing under-dispersion.  This is a restrictive type of negative feedback dependence and computation is challenging, taking several days to implement a single analysis.

In considering models that separate the marginal distribution from the dependence structure, it is natural to focus on copulas.  \citet{niko:karl:2010} proposed a copula model for bivariate counts that incorporates covariates into the marginal model. \citet{erha:czad:2009} proposed a copula model for high-dimensional counts, which can potentially allow under-dispersion in the marginals via a Faddy or Conway-Maxwell-Poisson \citep{shmu:etal:2005} model. \citet{gene:nesl:2007} provide a review of copula models for counts.  To our knowledge, copula models have not yet been developed that are directly applicable to count stochastic processes.  \citet{wils:ghah:2010} proposed a Gaussian copula process model to characterize dependence between arbitrarily many random variables independently of their marginals.  \citet{rodr:etal:2010} proposed a latent stick-breaking process, which is a nonparametric Bayes approach for a stochastic process with an unknown common marginal distribution modeled via a stick-breaking prior.  They considered a spatial count process application, with marginal modeled via a mixture of Poissons and the spatial dependence characterized through a latent Gaussian process.  This separates the marginal and dependence structure, but the marginal model is restrictive in being characterized as a mixture of Poissons, computation is intensive, and count functional data are not accommodated.

An alternative approach relies on rounding of a stochastic process. For classification it is common to threshold Gaussian process regression \citep{chu:ghah:2005, ghos:roy:2006}.  \citet{kach:yao:2009} rounded a real discrete autoregressive process to induce an integer-valued time series. \citet{cana:duns:2011} used rounding of continuous kernel mixture models to induce nonparametric models for count distributions.  This article instead proposes a class of stochastic processes that map a real-valued stochastic process $y^*: \mathcal{S} \to \Re$ to a count stochastic process $y: \mathcal{S} \to \mathcal{N}$.

	\section{Rounded Stochastic Processes }
\label{sec:roundedSP}

\subsection{Notation and model formulation}
\label{sec:model}

Let $y \in \mathcal{C}$ denote a count-valued stochastic process, with $\mathcal{S} \subset \Re^p$ compact and $\mathcal{C}$ the set of all $\mathcal{S} \to \mathcal{N}$ functions satisfying Assumption~\ref{ass:yregularity}.

\begin{assumption}
$y: \mathcal{S} \to \mathcal{N}$ is piecewise constant such that $\mathcal{S} = \bigcup_{l=1}^L \mathcal{S}_l(y)$, with $y(s)$ constant within the interior of each set $\mathcal{S}_l(y)$ and with unit increments at the boundaries $\mathcal{B}(y)$.  The boundary points fall within the set having the higher $y(s)$ value
\label{ass:yregularity}
\end{assumption} 

Assumption~\ref{ass:yregularity}  ensures that for sufficiently small changes in the input the corresponding change in the output is small.  We are particularly motivated by applications in which counts do not change erratically at nearby times but maintain some degree of similarity.  However, Assumption~\ref{ass:yregularity} does not rule out rapidly changing count processes, as one can have arbitrarily many jumps in a tiny interval and still satisfy the assumption. In addition, Assumption \ref{ass:yregularity} is easily relaxed.

We choose a prior $y \sim \Pi$, where $\Pi$ is a probability measure over $(\mathcal{C},\mathcal{B})$, with $\mathcal{B}(\mathcal{C})$ the Borel $\sigma$-algebra of subsets of $\mathcal{C}$.  The measure $\Pi$ induces the marginal probability mass functions
\begin{eqnarray}
	\mbox{pr}\{ y(s) = j \} = \Pi\{ y: y(s) = j \} = \pi_j(s),\quad j \in \mathcal{N},\quad s \in \mathcal{S}, 
\label{eq:pmf1}
\end{eqnarray}
and the joint probability mass functions
\begin{eqnarray}
	\mbox{pr}\{ y(s_1)=j_1, ... , y(s_k)=j_k \} = \Pi\{ y: y(s_1)=j_1, ... , y(s_k)=j_k\}
	= \pi_{j_1 ... j_k}(s_1, ... , s_k), 
\label{eq:pmf2}
\end{eqnarray}
for $j_h \in \mathcal{N}$ and $s_h \in \mathcal{S}$, $h=1,\ldots, k$, and any $k \geq 1$.

In introducing the Dirichlet process, \citet{art:ferg:1973} mentioned three appealing characteristics for nonparametric Bayes priors including large support, interpretability and ease of computation.  Our goal is to specify a prior $\Pi$ that gets as close to this ideal as possible.  Starting with large support, we would like to choose a $\Pi$ that allocates positive probability to arbitrarily small neighborhoods around any $y_0 \in \mathcal{C}$ with respect to an appropriate distance metric, such as $L^1$.  
To our knowledge, there is no previously defined stochastic process that satisfies this large support condition.  In the absence of prior knowledge that allows one to assume $y$ belongs to a pre-specified subset of $\mathcal{C}$ with probability one, priors must satisfy the large support property to be coherently Bayesian.  Large support is also a necessary condition for the posterior for $y$ to concentrate in small neighborhoods of any true $y_0 \in \mathcal{C}$. 

With this in mind, we propose to induce a prior $y \sim \Pi$ through
\begin{eqnarray}
	y = h(y^*),\quad y^* \sim \Pi^*, 
\label{eq:mapping}
\end{eqnarray}
where $y^*: \mathcal{S} \to \Re$ is a real-valued stochastic process, 
$h$ is a thresholding operator from $\mathcal{Y} \to \mathcal{C}$, $\mathcal{Y}$ is the set of all $\mathcal{S} \to \Re$ continuous functions, and $\Pi^*$ is a probability measure over  $(\mathcal{Y},\mathcal{B})$ with $\mathcal{B}( \mathcal{Y} )$ Borel sets.  
Unlike count-valued stochastic processes, there is a rich literature on real-valued stochastic processes.  For example, $\Pi^*$ could be chosen to correspond to a Gaussian process or could be induced through various basis or kernel expansions of $y^*$.

There are various ways in which the thresholding operator $h$ can be defined.  For interpretability and simplicity, it is appealing to maintain similarity between $y^*$ and $y$ in applying $h$, while restricting $y \in \mathcal{C}$.  Hence we focus on a rounding operator that let $y(s) = 0$ if $y^*(s) < 0$ and $y(s) = j$ if $j-1 \le y^*(s) < j$ for $j=1,\ldots,\infty$.  Negative values will be mapped to zero, which is the closest non-negative integer, while positive values will be rounded up to the nearest integer.   This type of restricted rounding ensures $y(s)$ is a non-negative integer.  Using a fixed rounding function $h$ in (\ref{eq:mapping}), we rely on flexibility of the prior $y^* \sim \Pi^*$ to induce a flexible prior $y \sim \Pi$.  For notational convenience and generality, we let
$y(s) = j$ if $y^*(s) \in A_j = [ a_j, a_{j+1} )$, with $a_0 < \cdots < a_{\infty}$ and we focus on $a_0 = -\infty, a_j = j-1, j=1,\ldots, \infty$.

This construction is particularly suitable for modeling dynamics of count processes close to zero and in particular, zero-inflated processes with local dependence in the zeros. Applying the mapping $h$ to a latent $y^*$  that assumes negative values across certain sub-regions of  $\mathcal S$ will lead to blocks of zeros in the count process $y$. This incorporates dependence between zero occurrences and the occurrence of small counts, which seems natural in most applications such as in the longitudinal tumor count study of \S \ref{sec:mouse}.

Figure~\ref{fig:examples1} illustrates the prior through showing realizations of the underlying stochastic process (Panel (a)) and resulting count process after applying the rounding operator (Panel (b)). The thick lines represents the mean functions of the real valued process and of the induced process. The latter is
\[
	E\{y(s)\} = \sum_{j=0}^{\infty} j \{F_{s}(a_{j+1}) -  F_{s}(a_j) \},
\]
where $F_s(x) = \int_{-\infty}^x f_s(y^*) \d y^*$ and $f_s$ is the marginal distribution of $y^*(s)$.

\begin{figure}
\centering
\subfigure{\includegraphics[scale=.9]{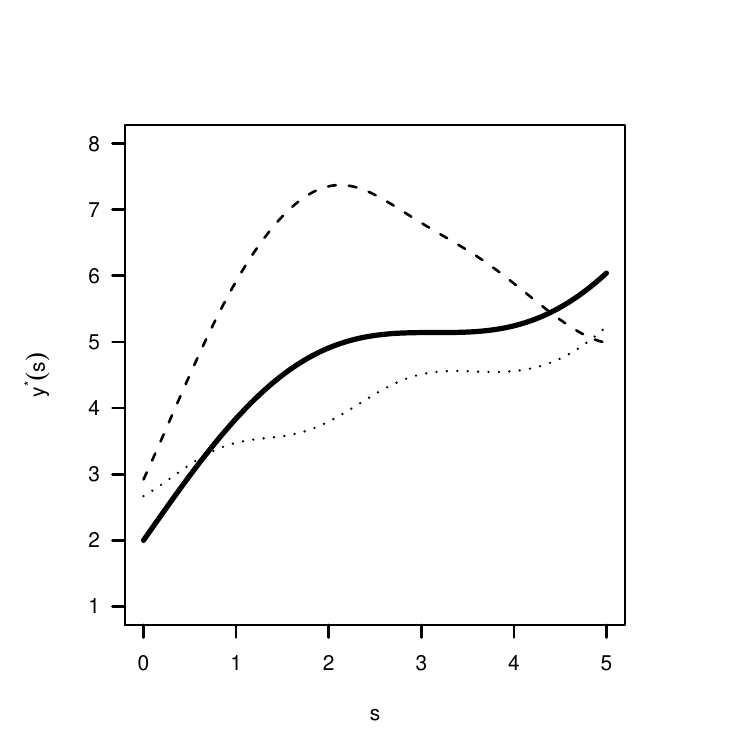}}
\subfigure{\includegraphics[scale=.9]{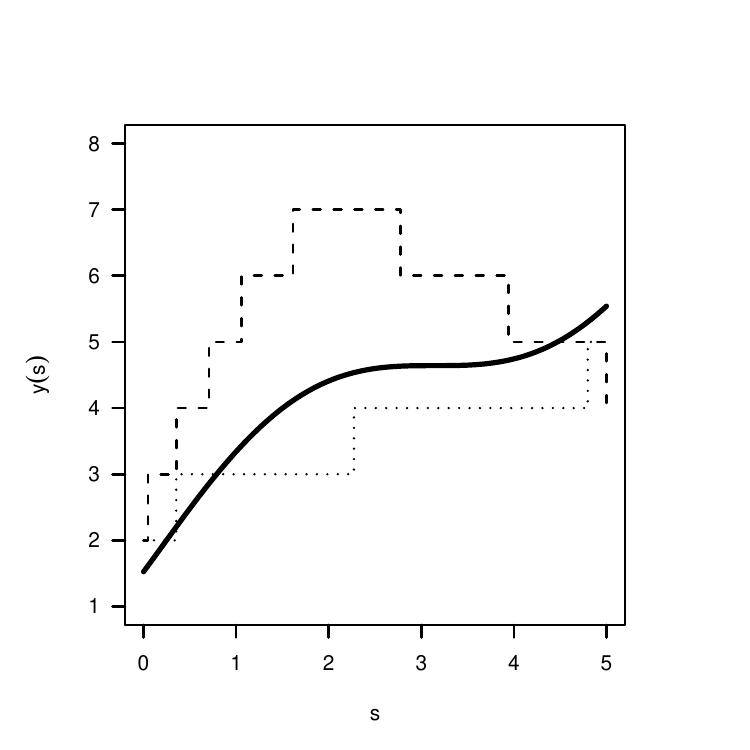}}
\caption{Panel (a) represents samples from a Gaussian process with mean function $\mu(s) = 2 + \sin(s) + s$ (bold line) and squared exponential covariance function. Panel (b) shows how the mapping operator (\ref{eq:mapping}) works. Dotted and dashed lines are the rounded version of panel (a) realizations while the bold line is the induced mean function.}
\label{fig:examples1}
\end{figure}

The covariance structure of the induced count process inherits much of the structure of the underlying process, as is clear from
\[
\mbox{cov}\{y(s),y(s')\} =  \sum_{j=0}^\infty \sum_{k=0}^\infty j\,k\, \mbox{pr}[ \{ y^*(s)\in A_{j},  y^*(s') \in A_k\} ] - E\{y(s)\}E\{y(s')\},
\]
where $\{y^*(s), y^*(s)\}$ has a bivariate distribution with covariance equal to cov$\{y^*(s), y^*(s)\}$.
We report some plots comparing the covariance of the original process with that of the induced process in the supplemental appendix.

In certain applications, count data can be naturally viewed as arising through integer-valued rounding of an underlying continuous process.  For example, in the longitudinal tumor count studies of \S \ref{sec:mouse}, it tends to be difficult to distinguish individual tumors and it is natural to posit a continuous time-varying tumor burden, with tumors fusing together and falling off over time.  In collecting the data, tumor biologists attempt to make an accurate count but measurement errors are unavoidable.  It is natural to accommodate this with a smoothly-varying continuous tumor burden specific to each animal with measurement errors and rounding producing the observed tumor counts.  However, even when there is no clear applied context motivating the existence of an underlying continuous process, the proposed formulation nonetheless leads to a highly flexible and computationally convenient model.


\subsection{Properties}
\label{subsec:properties}
The mapping function $h(\cdot)$ in (\ref{eq:mapping}) is many-to-one and the inverse mapping $h^{-1}(y)$  will correspond to an uncountable set of infinitely many continuous stochastic processes $y^*$ such that $y = h(y^*)$. As an important step in characterizing the support of the induced prior $y \sim \Pi$, Lemma~\ref{lem:existence} ensures the existence of at least one continuous stochastic process for each count process. All the proofs are reported in the Appendix.


\begin{lemma}
\label{lem:existence}
For every count stochastic process $y_0 \in \mathcal{C}$ satisfying Asssumption 1, there exists at least one continuous $y^*: \mathcal{S} \to \Re$ such that $y_0 = h(y^*)$.
\end{lemma}

Defining an $L^1$ neighborhood around $y_0$ of size $\epsilon$ as
\begin{equation}
	\eta_{\epsilon}(y_0) = \left\{ y: d_1(y_0,y) = \int |y_0(s)-y(s)| \d s < \epsilon  \right\},
\end{equation}
we state the following Theorem on the prior support.

\begin{theorem}
\label{theo:l1mapping}
Assuming the prior $\Pi^*$ assigns positive probability to $L^1$ neighborhoods of any continuous function $y_0^*: \mathcal{S} \to \Re$, the prior $\Pi$ induced through (\ref{eq:mapping}) assigns positive probability to $L^1$ neighborhoods of any $y_0 \in \mathcal{C}$ satisfying Assumption 1.
\label{theo:positivity}
\end{theorem}

In addition to showing large support of the prior, it is important to verify that the posterior distribution for $y$ concentrates increasingly around the true process $y_0$ as the sample size increases.  Theorem~\ref{theo:consistency} provides sufficient conditions under which $L^1$ posterior consistency is obtained. 
Assumption \ref{ass:infill} provides a space-filling regularity condition on the design.

\begin{assumption}
\label{ass:infill}
Let $\mathcal{S} = [0,1]^p$ and assume the $n$ values of $s_i$ arise from an in-fill design such that we can cover $\mathcal{S}$ with $n$ $L^\infty$ balls centered around $s_1, \dots, s_n$ of size $\delta$ with $2\delta \in \left( n^{-1/p},\lfloor n^{1/p}\rfloor^{-1} \right) $.
\end{assumption}

\begin{theorem}
\label{theo:consistency}
Let $y \in \mathcal{C}$ be a count stochastic process with $y_i = y(s_i)$, for $i=1,\ldots, n$ and $(s_1,\ldots, s_n)$ following Assumption~\ref{ass:infill}.
Letting $y_0 \in \mathcal{C}$ denote the true stochastic process and \mbox{$y \sim \Pi$}, then if $\Pi\{ \eta_{\epsilon}(y_0)\}>0$ for any $\epsilon$ and there exist sets $\{ \mathcal{C}_n \}_{n=1}^\infty$ with $\mathcal{C}_n \in \mathcal{C}$ and $\mathcal{C}_n^C$ the complement of $\mathcal{C}_n$, where $\Pi\{ \mathcal{C}_n^C \}< c_1 e^{-c_2 n}$, and $c_1, c_2$ positive constants, then
\begin{equation}
	 \Pi \left\{ \eta^C_{\epsilon}(y_0)\, |\, y_1,\ldots, y_n \right\} \to 0. 
\label{eq:posterior}
\end{equation}
\end{theorem}

From Theorems~\ref{theo:positivity} and \ref{theo:consistency}, it follows that the prior proposed in equation (\ref{eq:mapping}) will lead to $L^1$ posterior consistency under Assumptions~\ref{ass:yregularity}--\ref{ass:infill} as long as $\Pi^*$ assigns positive probability to $L^1$ neighborhoods of any continuous function and negligible probability  to $\mathcal{Y}_n^C = h^{-1}(\mathcal{C}_n^C)$ as $n$ increases. 
\citet{choi:sche:2007} showed that this condition holds, if $\mathcal{Y}_n^C$ has a particular form, for $\Pi^*$ corresponding to orthogonal basis expansions or Gaussian processes with continuously differentiable mean function and covariance having the form $k(s,s^{'};\beta) = k_0(\beta |s-s^{'}|)$, where $s \in \Re$, $k_0(s)$ is a positive multiple of a density function four times continuously differentiable on $\Re$ and suitable hyperprior is chosen for  $\beta$.

\subsection{Posterior computation}
\label{sec:computation}

We estimate the count process $y$ at locations $s^{(N)}$ including observed locations $s^{(n)} = (s_1,\ldots, s_n)^T$ and additional locations of interest $s_{n+1},\ldots,s_N$.  Our rounded Gaussian process lets $y^* \sim \mbox{\small{GP}}( 0, k)$, where $k(s,s') = \mbox{cov}\{ y^*(s), y^*(s') \} = \tau_1 \exp( -\tau_2 ||s-s' ||^2 )$ is the covariance with $\tau_1^{-1} \sim \mbox{Ga}(a_{\tau_1},b_{\tau_1})$ a scale parameter, $\tau_2^p \sim \mbox{Ga}(a_{\tau_2},b_{\tau_2})$ and $p$ the dimension of the domain $\mathcal{S}$.  Here, $\tau_2^{-1}$ is a bandwidth parameter controlling smoothness, and this prior is motivated by the optimality results of  \citet{vdva:vzel:2009}, though their theory does not apply directly in our setting.  The resulting joint distribution of $y^{*(n)} = \{ y^*(s_1),\ldots, y^*(s_n) \}^T$ is $N_n( 0, \Sigma_n)$, with $\Sigma_n = \{ \sigma_{ij} \}$ and $\sigma_{ij} = k(s_i,s_j)$.

Posterior computation can proceed via a Markov chain Monte Carlo algorithm.
\begin{itemize}
\item[Step]1 Sample $y^{*(n)}$ from $N_n(0, \Sigma_n)$ truncated to fall in a hyper-rectangle having
$a_{y_i} \le y^*(s_i) < a_{y_i+1}$.
\item[Step]2 Sample $\tau_1^{-1}$ from conditional posterior distribution $\mbox{Ga}(a_{\tau_1} + n/2, b_{\tau_1} + y^{*\T} \tau_1 \Sigma_n y^{*})$.
\item[Step]3 Update $\tau_2$ using a Metropolis-Hastings step.
\item[Step]4 After burn-in, sample $y^*(s_{n+1}),\ldots, y^*(s_N)$ from the multivariate Gaussian conditional distribution.
\end{itemize}
In Step 1, Gibbs sampling can be used to update each $y^*(s_i)$ from its univariate truncated Gaussian conditional, but this leads to slow mixing in our experience.  Instead, we use the slice sampler of \citet{liec:2010}, which samples multivariate Gaussians restricted to a rectangular region.  In step 3, the likelihood of $y^{(n)}$ marginalizing out $y^{*(n)}$ cannot be calculated analytically, so we rely on the multivariate normal likelihood of $y^{*(n)}$ in calculating the acceptance probability.  It is well known that updating $\tau_2$ conditionally on a latent Gaussian process can lead to stickiness, but due to the fact that our rounding approach minimizes differences between the observed $y$ and the latent $y^*$ we have not found this to be a major problem.  Alternatively, one can improve mixing using the slice sampling approach of  \citet{murr:adam:2010} with some additional complexity.

As for other Gaussian process models, we face a computational bottleneck and numerical instability as we evaluate $y^*$ at increasing numbers of locations. Particularly when the process is observed at close locations and the covariance function favors smooth realizations, one obtains
an ill-conditioned matrix, which can lead to large computational errors which degrade performance. There is a rich literature proposing solutions, with \citet{bane:etal:2012} a recent example. A widely-used approximation represents the function as a linear combination of finitely many basis functions, leading to reduced instability problems and potentially improving Markov chain Monte Carlo mixing. 
Hence, along with the rounded Gaussian process, we implement an alternative that approximates $y^*$ using penalized splines, with details on this approach provided in a supplemental appendix.

\section{Simulation study}

A simulation study is conducted to assess the performance of the proposed approach, implemented using rounded Gaussian processes or P-splines, relative to several competitors. The first set of competitors initially treats the count measurements as continuous, assuming $h$ to be the identity function.  The estimated continuous trajectory is then rounded in a second stage to produce an estimated count process.  Such ad hoc two-stage approaches are simple to implement; we consider two-stage versions of rounded Gaussian processes and P-splines.  A second approach treats the count measurements as ordered categorical using the Gaussian process ordinal regression model of \citet{chu:ghah:2005}.  This method faces complications when applied to counts and sparse ordered categorical data.  In particular, letting $y(s_i) \in \{0,1,\ldots, d\}$ for $i=1,\ldots, n$ and $n_j = \sum_{i=1}^n 1_{\{y(s_i)=j\}}$, the total number of observations having value $j$,  poor performance was obtained when any $n_j$ was small, with lack of convergence when $n_j=0$ for any $j \in \{0,1,\ldots,d\}$.  A third approach corresponds to Poisson regression with mean parameter $\lambda(s)$ estimated with a spline smoother as done by default by gam function of R library {\small MASS}.  Lastly, we consider a simple interpolating step function defined as 
\begin{equation}
f(s) = y_1 1_{s < s_{2} }(s) + \sum_{j=2}^n y_j 1_{s_j \leq s < s_{j+1} }(s).
\end{equation}
For our method, we considered the posterior median of $y(s)$.

Simulations have been run under a wide variety of settings leading to qualitatively similar results. We report the results for four scenarios. The first scenario generates count stochastic processes from Poisson$\{\lambda(s)\}$, with  
$\lambda(s)= 2+ s/5 + \sin(s)$.  In the second scenario, $y$ is generated by rounding a realization of a Gaussian process plus an error term,
\begin{eqnarray}
	y = h(y^*), \, \, \, \, \, \, y^* \sim \text{\small{GP}} (\mu, k) + \epsilon,  \label{eq:simcase2}
\end{eqnarray}
with mean function $\mu(s) = 2 + \exp (s / 5)$, covariance function $k(s,s^{'})$ squared exponential and $\epsilon(s)$ independent draws from $N(0,2)$. These two cases do not satisfy Assumption 1, since infinitely many discontinuity points can occur. Under the third scenario, we generate from a Poisson count process with rate parameter $1/2$ and in the fourth from (\ref{eq:simcase2}) with $\epsilon = 0$.

For each case, we generated data on a equispaced grid of $1,000$ points between $0$ and $20$. 
Taking equispaced subsamples for different level of sparsity, namely of sizes $n=25$, $n=50$, $n=100$, and $n = 500$, we estimate the trajectory on a fine grid for 500 replicates for each scenario and each method. 
Using Markov chain Monte Carlo, we obtained draws from the posterior predictive distribution and used the median as our estimate. Methods are compared based on averaging the mean absolute deviation between the estimate and the true process across the replicates and grid points.

\begin{sidewaystable}								
\def~{\hphantom{0}}								
\caption{Mean absolute deviation (and standard deviation) in simulation study of Section \ref{sec:computation}. RGP, rounded Gaussian process; GP, Gaussian process; RPS, rounded P-splines; PS, P-spline; GPOR, Gaussian process ordinal regression;
 NPP, nonparametric Poisson model; E, empirical interpolating step function.								
}{%
\begin{tabular}{lcccccccc}								
 \\								
& \multicolumn{4}{c}{Scenario 1} & \multicolumn{4}{c}{Scenario 2} \\								
&	$n = 25$ &	$n = 50$ &	$n = 100$ &	$n = 500$ &	$n = 25$ &	$n = 50$ &	$n = 100$ &	$n = 500$ \\
RGP & 	2$\cdot$10 (0$\cdot$17) & 	2$\cdot$04 (0$\cdot$11) & 	1$\cdot$93 (0$\cdot$09) &	1$\cdot$02 (0$\cdot$05) & 	2$\cdot$27  (0$\cdot$03) & 	2$\cdot$12 (0$\cdot$02) & 	1$\cdot$98 (0$\cdot$01) & 	0$\cdot$93 (0$\cdot$01) \\
GP &	2$\cdot$12 (0$\cdot$02) & 	2$\cdot$07 (0$\cdot$02) &	1$\cdot$98 (0$\cdot$01) &	1$\cdot$05 (0$\cdot$01) &	2$\cdot$27  (0$\cdot$03) & 	2$\cdot$13 (0$\cdot$02) & 	1$\cdot$99 (0$\cdot$01) & 	0$\cdot$95 (0$\cdot$01) \\
RPS &	1$\cdot$70 (0$\cdot$09) &	1$\cdot$62 (0$\cdot$07) &	1$\cdot$5 (0$\cdot$05) &	0$\cdot$79 (0$\cdot$03) &	1$\cdot$78 (0$\cdot$11) &	1$\cdot$65 (0$\cdot$06) &	1$\cdot$51 (0$\cdot$05) &	0$\cdot$81 (0$\cdot$03) \\
PS &	1$\cdot$70 (0$\cdot$08) &	1$\cdot$63 (0$\cdot$07) &	1$\cdot$51 (0$\cdot$05) &	0$\cdot$8 (0$\cdot$03) &	1$\cdot$81 (0$\cdot$13) &	1$\cdot$69 (0$\cdot$07) &	1$\cdot$55 (0$\cdot$06) &	0$\cdot$83 (0$\cdot$03) \\
GPOR &	2$\cdot$26 (0$\cdot$33) &	2$\cdot$22 (0$\cdot$26) &	2$\cdot$14 (0$\cdot$21) &	2$\cdot$18 (0$\cdot$14) &	2$\cdot$47 (0$\cdot$31) &	2$\cdot$46 (0$\cdot$31) &	2$\cdot$42 (0$\cdot$24) &	2$\cdot$73 (0$\cdot$14) \\
NPP &	1$\cdot$74 (0$\cdot$08) &	1$\cdot$69 (0$\cdot$06) &	1$\cdot$66 (0$\cdot$05) &	1$\cdot$64 (0$\cdot$04) &	1$\cdot$78 (0$\cdot$1) &	1$\cdot$71 (0$\cdot$06) &	1$\cdot$69 (0$\cdot$06) &	1$\cdot$66 (0$\cdot$05) \\
E &	2$\cdot$2 (0$\cdot$18) &	2$\cdot$19 (0$\cdot$13) &	2$\cdot$2 (0$\cdot$1) &	2$\cdot$2 (0$\cdot$08) &	2$\cdot$58 (0$\cdot$2) &	2$\cdot$31 (0$\cdot$13) &	2$\cdot$25 (0$\cdot$11) &	2$\cdot$21 (0$\cdot$06) \\
& \multicolumn{4}{c}{Scenario 3} & \multicolumn{4}{c}{Scenario 4} \\								
&	$n = 25$ &	$n = 50$ &	$n = 100$ &	$n = 500$ &	$n = 25$ &	$n = 50$ &	$n = 100$ &	$n = 500$ \\
RGP & 	0$\cdot$12 (0$\cdot$01) & 	0$\cdot$07 (0$\cdot$01) & 	0$\cdot$05 (0$\cdot$01) &	0$\cdot$01 (0$\cdot$01) &	0$\cdot$34 (0$\cdot$01) & 	0$\cdot$25 (0$\cdot$01) & 	0$\cdot$17 (0$\cdot$01) &	0$\cdot$05 (0$\cdot$01) \\
GP &	0$\cdot$42 (0$\cdot$01) &	0$\cdot$37( 0$\cdot$01) & 	0$\cdot$34 (0$\cdot$01) & 	0$\cdot$21 (0$\cdot$01) &	0$\cdot$58 (0$\cdot$06) &	0$\cdot$52 (0$\cdot$04) & 	0$\cdot$47 (0$\cdot$02) & 	0$\cdot$26 (0$\cdot$01) \\
RPS &	0$\cdot$14 (0$\cdot$06) &	0$\cdot$08 (0$\cdot$04) &	0$\cdot$05 (0$\cdot$03) &	0$\cdot$02 (0$\cdot$01) &	0$\cdot$28 (0$\cdot$07) &	0$\cdot$19 (0$\cdot$05) &	0$\cdot$13 (0$\cdot$03) &	0$\cdot$05 (0$\cdot$01) \\
PS &	0$\cdot$41 (0$\cdot$08) &	0$\cdot$39 (0$\cdot$06) &	0$\cdot$37 (0$\cdot$06) &	0$\cdot$21 (0$\cdot$03) &	0$\cdot$53 (0$\cdot$06) &	0$\cdot$49 (0$\cdot$04) &	0$\cdot$46 (0$\cdot$03) &	0$\cdot$26 (0$\cdot$01) \\
GPOR &	2$\cdot$88 (0$\cdot$8) &	3$\cdot$03 (1$\cdot$04) &	3$\cdot$26 (1$\cdot$4) &	3$\cdot$65 (1$\cdot$51) &	2$\cdot$25 (2$\cdot$04) &	2$\cdot$6 (3$\cdot$97) &	4$\cdot$74 (8$\cdot$52) &	5$\cdot$9 (10$\cdot$09) \\
NPP &	0$\cdot$27 (0$\cdot$09) &	0$\cdot$26 (0$\cdot$09) &	0$\cdot$26 (0$\cdot$09) &	0$\cdot$26 (0$\cdot$09) &	0$\cdot$56 (0$\cdot$12) &	0$\cdot$56 (0$\cdot$12) &	0$\cdot$56 (0$\cdot$12) &	0$\cdot$56 (0$\cdot$12) \\
E &	0$\cdot$18 (0$\cdot$07) &	0$\cdot$09 (0$\cdot$04) &	0$\cdot$05 (0$\cdot$02) &	0$\cdot$01 (0) &	1$\cdot$11 (0$\cdot$06) &	0$\cdot$59 (0$\cdot$05) &	0$\cdot$31 (0$\cdot$03) &	0$\cdot$09 (0$\cdot$01) \\
\end{tabular}}								
\label{tab:simulation1}								
\end{sidewaystable}

From Table~\ref{tab:simulation1}, it is apparent that the proposed rounding approaches have the best overall performance.  The Gaussian process ordinal regression model consistently has the worst performance.  As expected the Poisson model with nonparametric mean performs well in scenario 1 but poorly in  other cases, particularly when the sample size is not small.  The interpolating step function has consistently poor performance except in scenario 3.  The two stage methods perform similarly to the proposed approaches in scenarios 1 and 2, but have substantially worse performance in scenarios 3 and 4. The two stage methods have particularly poor performance when counts do not take a wide range of values, have values near zero, or tend to have many occurrences of the same value. In addition, the approach of rounding in a second stage can have unanticipated consequences in terms of inference on functionals, which may be unreliable and biased.
Interestingly, the rounded P-splines approach has somewhat better performance than the rounded Gaussian process.  Since rounded P-splines are also faster to implement, taking from 15 seconds for samples of size $n = 25$ to 30 seconds for samples of size $n = 500$ for 10,000 MCMC iterations in each of the simulated examples, we focus on this approach in the real data applications.
We also compared the methods in terms of predictive mean absolute deviation, width and coverage of predictive credible intervals and again observed better performance overall for the proposed approaches, with the competitors having high mean absolute deviation and poor coverage in at least one of the cases.  Additional tables summarizing the results for predictive errors and predictive coverage are reported in the supplemental appendix.

\section{Real data application}

\subsection{Count functional data}
\label{sec:countfda}

We have focused on the case in which there is a single count process $y$ observed at locations $s = (s_1,\ldots, s_n)^T$.  In many applications, there are instead multiple related count processes $\{ y_i, i=1,\ldots, n \}$, with the $i$th process observed at locations $s_i = (s_{i1},\ldots, s_{in_i})^T$.  We refer to such data as count functional data.  As in other functional data settings, it is of interest to borrow information across the individual functions through use of a hierarchical model.  This can be accomplished within our rounded stochastic processes framework by first defining a functional data model for a collection of underlying continuous functions $\{ y_i^*, i=1,\ldots, n\}$, and then letting $y_i = h(y_i^*)$, for $i=1,\ldots, n$.  There is a rich literature on appropriate models for $\{ y_i^*, i=1,\ldots, n\}$ ranging from hierarchical Gaussian processes \citep{art:behs:etal:2005} to wavelet-based functional mixed models \citep{art:morr:carr:2006}. 

Let $y_i(s)$ denote the count for subject $i$ at time $s$, $y_{it} = y_i(s_{it})$, where $s_{it}$ is the $t$th observation time for subject $i$, and $x_{it} = (x_{it1},\ldots,x_{itp})^T$ predictors for subject $i$ at the $t$th observation time.  As a simple model motivated by the longitudinal tumor count and asthma inhaler use applications described below, we let 
\begin{eqnarray}
y_{it} = h( y_{it}^* ),\quad y_{it}^* = \xi_i + b(s_{it},x_{it})^T\theta + \epsilon_{it},\quad \xi_i \sim Q,\quad \epsilon_{it} \sim N(0,\tau^{-1}), \label{eq:func1}	
\end{eqnarray}
where $\xi_i$ is a subject-specific random effect, $b(\cdot)$ are basis functions that depend on time and predictors, $\theta$ are unknown basis coefficients, and $\epsilon_{it}$ is a residual which allows the counts to vary erratically from time to time about the smooth subject-specific mean curve.
We use basis expansions motivated by the success of rounded P-splines in our simulation.  To allow the random effect distribution to be unknown, we choose a Dirichlet process prior \citep{art:ferg:1973}, with $Q \sim \mbox{\small{DP}}(\alpha Q_0)$, with $\alpha$ a precision parameter and the base measure $Q_0$ chosen as $N(0,\psi)$ with $\psi \sim \mbox{Ga}(a_\psi, b_\psi)$.  As commonly done we fix $\alpha =1$.
We additionally choose a hyperprior for the residual precision $p(\tau) \propto \tau^{-1}$ and for the basis coefficients $p(\theta)$, with the specific form of $p(\theta)$ depending on the context.

\subsection{Transgenic mouse bioassay application}
\label{sec:mouse}

We first analyze data from a Tg.AC mouse bioassay study of pentaerythritol triacrylate, a chemical used in many industrial processes. Animals are randomized to a control or one of five dose groups each of size 30. The five dose groups are 0$\cdot$75, 1$\cdot$5, 3, 6, or 12 mg/kg. The number of skin papillomas on the back of each mouse is counted weekly for 26 weeks and it is of interest to compare the groups to see if there is an increase in tumorigenicity relative to control, while assessing dose response trend.  \citet{duns:herr:2005} analyzed these data through a Poisson-gamma frailty model.  As motivated in \S 1$\cdot$2, Poisson hierarchical models are quite restrictive and our focus here is on using the proposed model to improve robustness.

The only predictor for an animal is the dose group $x_i \in \{1,\ldots,G\}$ and we let $b(s_{it},x_i)^T\theta = b(s_{it})^T\theta_{x_i}$ in expression (\ref{eq:func1}) to allow a separate trajectory in time for each dose group,
with $b(s)$ B-spline basis functions, $\theta_g$ basis coefficients specific to group $g$, and $p(\theta_g\, |\, \lambda) \propto \exp(-1/2 \lambda \theta_g^T P \theta_g)$  conditionally independent P-spline priors for each dose group. 
The prior is designed to only borrow information across dose groups in estimating smoothness parameter $\lambda$ to avoid the possibility of having chemical effects in higher dose groups pull up the estimated tumor response in lower dose groups. To induce a heavy-tailed prior having appealing computational properties, we use a multilevel hierarchical prior for $\lambda$, with 
\mbox{$\lambda \sim \text{Ga}(\nu/2,\delta\nu/2)$,} $\delta \sim \text{Ga}(a_\delta, b_\delta)$ and $\psi \sim \text{Ga}(a_\psi, b_\psi)$. 
We do not expect to have substantial learning from the data about $\delta$ or $\psi$. Computational details are reported in a supplemental appendix.

As a global measure of toxicity, we use the average papilloma burden per group. The two lower dose groups showed no significant difference from the control group with the posterior mean  of the average tumor burden $<$0$\cdot$001 and the 95\% credible intervals concentrated near zero. 
In the higher groups the average tumor burden grows with the dose level. Mean tumor burden and 95\% credible intervals are  
0$\cdot$18 [0$\cdot$06,0$\cdot$39], 9$\cdot$51 [9$\cdot$21,9$\cdot$80] and 12$\cdot$33 [11$\cdot$90,12$\cdot$72] 
for the 3, 6 and 12 mg/kg dose group respectively. Cumulative tumor burdens along with the dose group-specific empirical means for each week are reported in 
Figure~\ref{fig:cumulativetumor}.
\begin{figure}
\centering
\includegraphics[scale=.8]{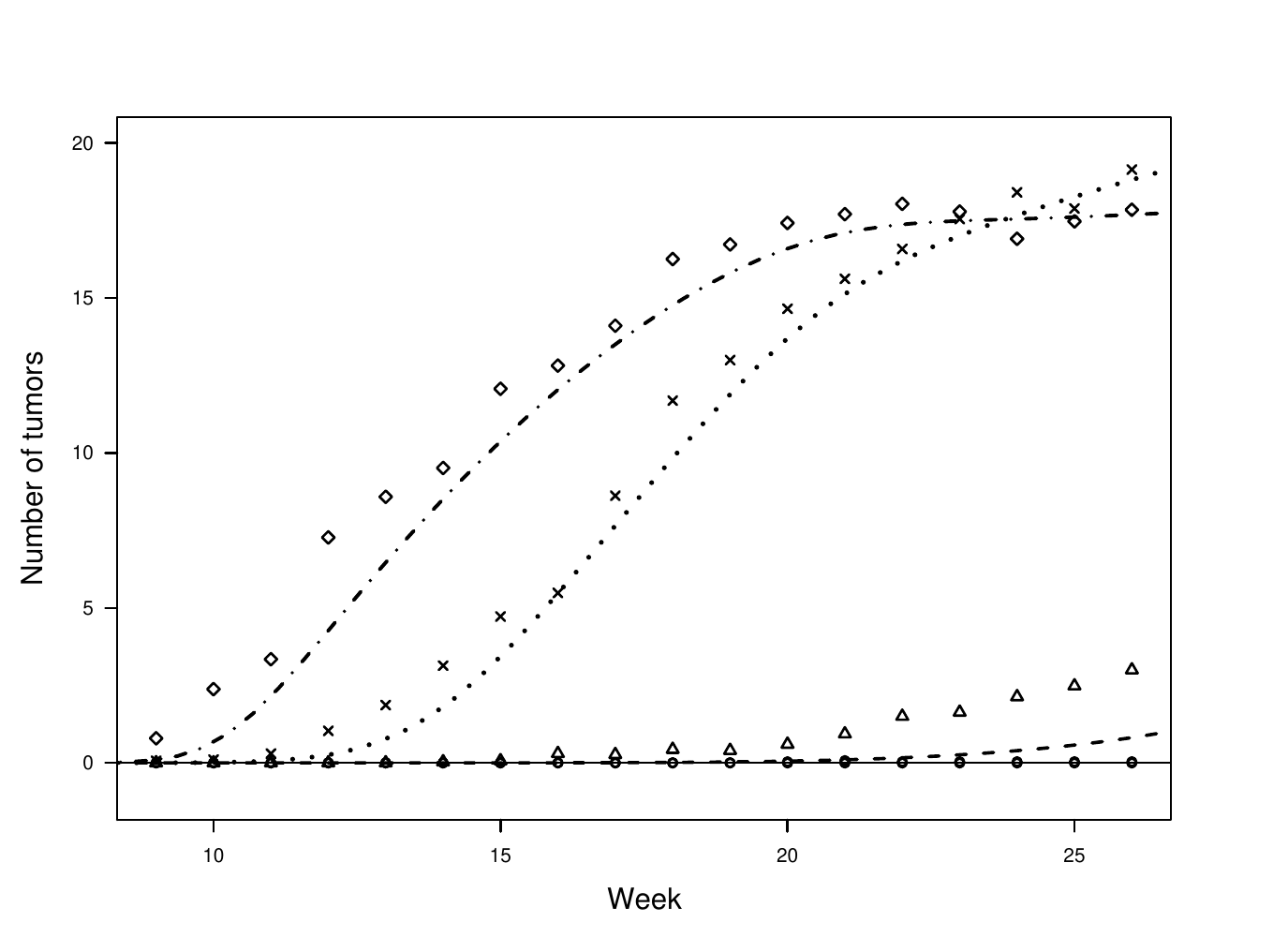}
\caption{Estimated cumulative mean tumor burden (lines) and weekly sample means (points) for the control group, 0$\cdot$75 mg/kg and 1$\cdot$5 mg/kg (solid line and circles), 3 mg/kg (dashed line and triangles), 6 mg/kg (dotted line and crosses) and 12 mg/kg (dash-dotted line and squares) dose groups.}
\label{fig:cumulativetumor}
\end{figure}

As a measure of time varying increase in papilloma burden, we computed the mean burden per dose group per week subtracting the average number for the control group. Posterior means and 95\% credible bands are reported in Figure~\ref{fig:tveff}. 
The two lower dose groups are indistinguishable from control, with panel (a) of Figure~\ref{fig:tveff} being a straght line equal to zero, while the 3, 6 and 12 mg/kg dose groups exhibit clear increases relative to control starting from the 17th, 9th and 8th week, respectively.
\begin{figure}
\centering
\includegraphics[scale=.9]{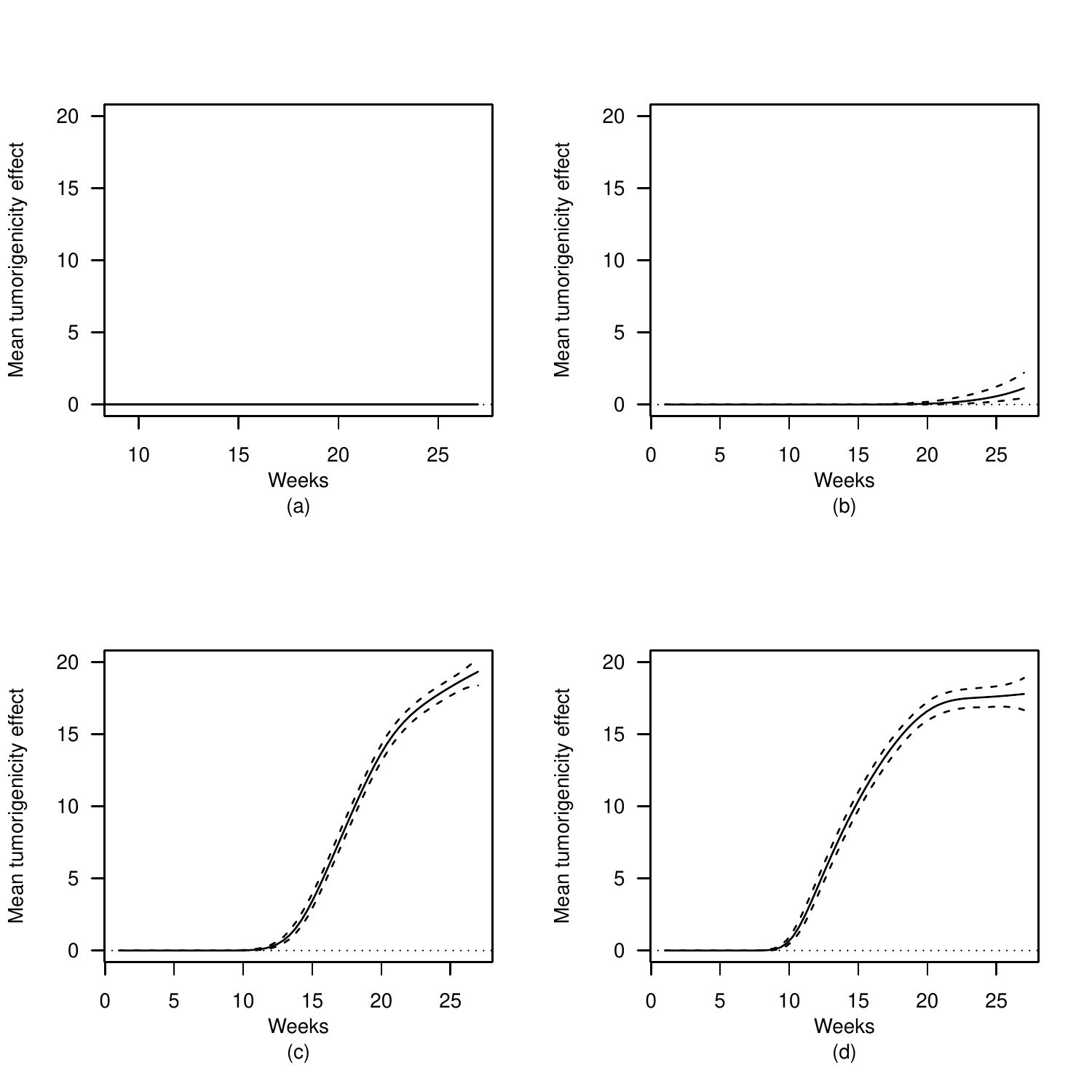}
\caption{Time varying chemical exposure posterior mean effect on tumorigenicity (continuous line) with 95\% credible bands (dashed lines) for (a) 0$\cdot$75 mg/kg and 1$\cdot$5 mg/kg, (b) 3 mg/kg, (c) 6 mg/kg and (d) 12 mg/kg dose groups. Dotted line at zero corresponds to no effect of the chemical.}
\label{fig:tveff}
\end{figure}
Higher dosages lead to higher numbers of skin papillomas, and earlier onset of the first tumor. Our modeling approach allows us to estimate the average
time of onset of first tumor, which occurs on the 27th, 14th and 11th week, for the three higher dose groups.  In other groups, the typical mouse did not develop tumors prior to the end of the study. 

Our overall conclusions agree with \citet{duns:herr:2005}, though the estimates differed somewhat. The group comparison results were also consistent with results from a frequentist generalized linear model analysis.  We additionally implemented standard frequentist nonparametric tests for comparing groups based on summaries of the tumor trajectory data including time of first tumor and maximum tumor burden per animal.   A p-value less than 0$\cdot$001 for the Kruskal-Wallis rank sum test suggested strong evidence against equality among the dose groups in the maximum tumor burden per animal.  Pairwise Wilcoxon tests were performed to test the equality of the maximum burden between each treated group and control, with one-sided alternatives of higher maximum burdens in the treated groups. P-values were less than 0$\cdot$01 for the three higher groups while being 1 and 0$\cdot$09 for the 0$\cdot$75 and 1$\cdot$5 mg/kg groups respectively. 
Similar results are obtained considering the time of development of the first tumor as a summary of the tumor trajectory.
As partly illustrated in Figure~\ref{fig:cumulativetumor}, which shows the empirical and estimated mean tumor burdens in each group, the model has a good fit to the data.

\subsection{Asthma inhaler use application}

We additionally analyzed data on daily usage of albuterol asthma inhalers \citep{grun:bruc:2011}. Daily counts of inhaler use were recorded for a period between 36 and 122 days for 48 students previously diagnosed with asthma. The total number of observations was 5,209.  As discussed by  \citet{grun:bruc:2011}, the data are under-dispersed.  Let $y_{it}$ denote the number of times the $i$th student used the inhaler on day $t$.  Interest focuses on the impact of morning levels of  {\small PM}$_{25}$, small particles less than 25 mm in diameter in air pollution, on asthma inhaler use.  At each day $t$, a vector $x_t = (x_{t1},\ldots,x_{tp})^T$ of environmental variables are recorded including {\small PM}$_{25}$, average daily temperature (Fahrenheit degree/100), \% humidity and barometric pressure (mmHg/1000).  We modify (\ref{eq:func1}) to include these predictors in an additive model as follows.
\begin{eqnarray}
y_{it} = h( y_{it}^* ), \quad y_{it}^* = \xi_i  + \sum_{j=1}^4 b_j(x_{jt})^T \theta_j + \epsilon_{it}, \label{eq:asthma}
\end{eqnarray} 
where $\xi_i$ is a random effect modeled as in \S 4$\cdot$2, $b_j$ is a  B-spline basis with $\theta_j$ the basis coefficients and $\epsilon_i \sim N(0, \tau^{-1} R)$, with $R$ the correlation matrix arising from a first order autoregressive process with correlation parameter $\rho$. The prior for each $\theta_j$ is identical to the prior used for $\theta_g$ in \S 4$\cdot$2 and each predictor is normalized to have mean zero and unit variance prior to analysis.  The correlation parameter is given a uniform prior on $[-1,1]$.  Computational details are reported in a supplemental appendix.

We ran our Markov chain Monte Carlo algorithm for 10,000 iterations with a 1,000 iteration burn-in discarded.  Convergence and mixing were diagnosed by monitoring the non-linear effects of the different predictors at several values and also monitoring hyperparameters; 
The trace plots showed excellent mixing, with effective sample size over 9,000. Autocorrelation functions tend to drop near zero between lag 1 and 2.
To obtain interpretable summaries of the non-linear covariate effects on the inhaler use counts, we recorded for each predictor at a dense grid of $x_{jt}$ values at each sample after burn-in the conditional expectation of the count for a typical student  having $\xi_i = \mu_Q$, where $\mu_Q$ is the mean of the random effects distribution $Q$,
\begin{eqnarray}
\mu_j(x_{jt}) & = & \mbox{E}( y_{it}\, |\, x_{jt}, x_{j't}=0, j' \neq j, \xi_i=\mu_Q, \theta, \tau, \rho) \notag \\
	      & \approx & \sum_{k=0}^{\lceil K \rceil} k [\Phi\{a_{k+1}; \mu_j^*(x_{jt}) , \tau\} -  \Phi\{a_k; \mu_j^*(x_{jt}),\tau\}],   
\label{eq:meaneff}
\end{eqnarray}
where $\Phi(\cdot;\mu,\tau)$ is the cumulative distribution function of a normal random variable with mean $\mu$ and precision $\tau$, $K$ is the 99$\cdot$99\% quantile of $N\{ \mu_j^*(x_{jt}) , \tau^{-1}\}$, and 
\begin{equation}
	\mu_j^*(x_{jt}) = b_j(x_{jt})^\T \theta_j + \sum_{l\neq j} b_l(0)^\T \theta_l + \mu_Q,
\end{equation}
with the other predictors fixed at their mean value. Based on these samples, we calculated posterior means and pointwise 95\% credible intervals, with the results reported in Figure~\ref{fig:asthma1}.  

These data were previously analyzed by \citet{grun:bruc:2011} using a Faddy distribution with a log-linear mixed model for the mean, 
\begin{equation}
\log \mbox{E}(y_{it}\, |\, x_t, \beta, u_i, e_{it}) =  \sum_{j=1}^p x_{jt} \beta_j + u_i + e_{it},
\end{equation} 
where $u_i$ is a subject-specific random effect and $e_{it}$ is a residual following a first-order autoregressive process.  They estimated a coefficient of 0$\cdot$013 for {\small PM}$_{25}$, which is close to zero with 95\% confidence interval including zero.  
A Poisson log-linear model analysis yielded a similar coefficient of 0$\cdot$014 but with a 50\% wider confidence interval.
Our approach, which is based on a substantially more flexible model that allows nonlinear effects and a nonparametric random effects distribution, produces results that are consistent with these earlier analyses.

\begin{figure}
\centering
\includegraphics[scale=.8]{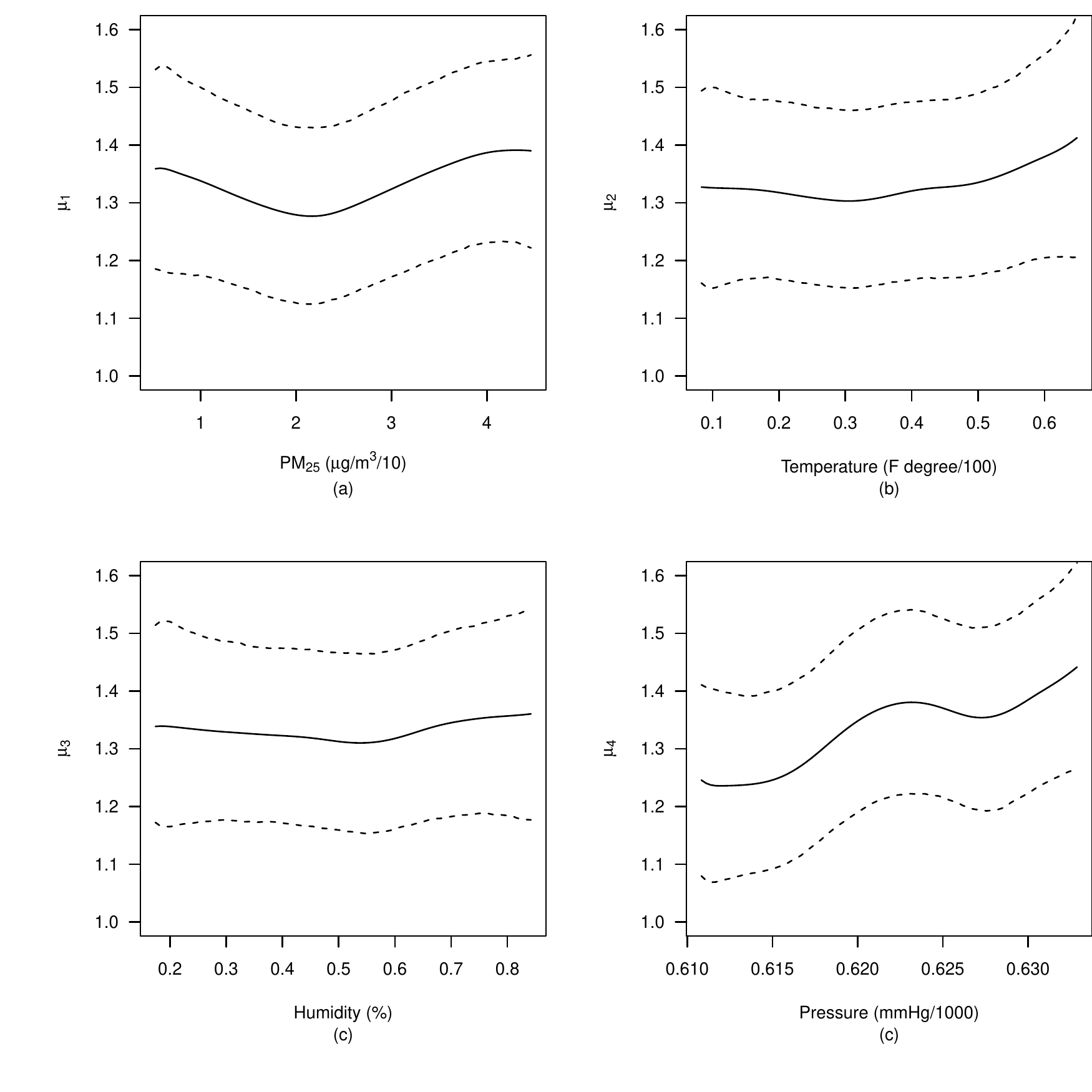}
\caption{Posterior mean and 95\% pointwise credible bands for the effect of (a) concentration of {PM$_{25}$} pollutant, (b) average daily temperature, (c) \% of humidity, and (d) barometric pressure on asthma inhaler use calculated with equation (\ref{eq:meaneff}).}
\label{fig:asthma1}
\end{figure}

\section{Discussion}

We have proposed a simple new approach for modeling count stochastic processes based on rounding continuous stochastic processes.  The general strategy is flexible and allows one to leverage existing algorithms and code for posterior computation for continuous stochastic processes.  Although rounding of continuous underlying processes is quite common for binary and categorical data, such approaches have not to our knowledge been applied to induce new families of count stochastic processes.  Instead, the vast majority of the literature for count processes relies on Poisson process and hierarchical Poisson constructions, which have some well known limitations in terms of flexibility.  We have explored some basic properties of rounding, but the primary contribution of this article is introducing the idea that rounding is useful in this context and we hope to stimulate additional work on properties of the resulting processes.  It is likely that some properties of the underlying continuous process, which are well known for Gaussian processes and in other standard cases, may carry over to the induced count process.   However, this deserves further study. There are also interesting directions in terms of modeling counting processes corresponding to non-decreasing count processes via rounding non-decreasing continuous processes using monotone splines \citep{rams:1998,neel:duns:2004,shiv:etal:2009} and other constructions.

\section*{Acknowledgements}
This research was partially supported by grants from the National Institute of Environmental
Health Sciences of the United States National Institutes of Health and the University of Padua in
Italy.

\appendix

\section*{Appendix 1}

\begin{proof}[of Lemma 1]

For any count stochastic process $y_0$ satisfying Assumption 1, we can partition the domain $\mathcal{S}$ into mutually disjoint sets $\mathcal{S}_l(y_0)$, with $y_0(s)$ constant within the interior of each $\mathcal{S}_l(y_0)$ and with unit increments at the boundaries.  There are clearly infinitely many continuous functions $y^*: \mathcal{S} \to \Re$ satisfying the constraints
(i) $y^*(s) \in [a_{y_0(s)},a_{y_0(s)+1})$ for all $s \in \mathcal{S}$ and
(ii) $y^*(s) = a_{y_0(s)}$ for $s \in \mathcal{B}(y_0)$.
For all such $y^*$, we have $y_0 = h(y^*)$.
\end{proof}

\begin{proof}[of Theorem \ref{theo:l1mapping}]
Theorem is an immediate consequence of Lemma \ref{lem:existence} and of the following Lemma ensuring that the mapping $h$ maintains $L^1$ neighborhoods.

\begin{lemma}
\label{lem:neigh}
Suppose $y^*$ and $y_0^*$ are continuous and bounded by $M \in \Re$ such that $d_1(y^*, y_0^*) = \epsilon^*$, $y=h(y^*)$ , $y_0 = h(y_0^*)$. Then, $y \in \eta_{\epsilon}(y_0)$ for all $\epsilon > \zeta(\epsilon^*;y^*_0)$, where $\zeta(\epsilon^*;y_0^*)$ is non decreasing in $\epsilon^*$ having $\lim_{\epsilon^* \to 0} \zeta(\epsilon^*;y^*_0) = 0$.
\end{lemma}

\begin{proof}[of Lemma \ref{lem:neigh}]
Take $\mathcal S = [0,1]^p$ without loss of generality. Let $\{ \mathcal{S}_l(y_0,y) \}_{l =1}^m$ the partition of $\mathcal S$ induced by $\{ \mathcal{S}_l(y_0)\}_{l =1}^{m_0}$ and $\{ \mathcal{S}_l(y)\}_{l =1}^{m_1}$ such that $y(s) = j_l$ and $y_0(s) =k_l$ for all $s \in \mathcal{S}_l(y_0,y)$ and some $j_l,k_l \in \mathcal{N}$. 
Let $\delta_l(y_0,y) = |j_l - k_l|$, for $l= 1, \dots, m$ and $\lambda(\cdot)$ be the Lebesgue measure. Define
\[
	\zeta(\epsilon^*;y^*_0) = \sup_{y^* \in \eta_{\epsilon^*} (y_0^*) } \left\{ \max_{l=1,2,\dots} \left[ \delta_l\{ y_0,h(y^*) \} \right]  \sum_{l: \delta_l \neq 0} \lambda[ \mathcal{S}_l\{y_0,h(y^*) \} ]  \right\}.
\]
Clearly $y \in \eta_{\epsilon}(y_0)$ for all $\epsilon > \zeta(\epsilon^*;y_0^*)$  since 
\[
	d_1(y_0,y) = \sum_{l=1}^m \delta_l(y_0,y)	 \lambda\{\mathcal{S}_l(y_0,y) \} \leq \zeta(\epsilon^*;y_0^*).
\]
We show first that $\lim_{\epsilon^* \to 0} \zeta(\epsilon^*;y_0) = 0$. What follows holds for all $y^* \in \eta_{\epsilon^*} (y_0^*)$. Consider the general $y^* \in \eta_{\epsilon^*} (y_0^*)$. Since $\sum_{l: \delta_l \neq 0} \lambda[ \mathcal{S}_l\{y_0,h(y^*) \} ] $ is finite, $\zeta(\epsilon^*;y_0)$ goes to zero if $\max \delta_l\{y_0,h(y^*)\}$ goes to zero.
Define $M_{\epsilon^*} = \max  \left| y^*(s) - y_0^*(s) \right|$ and let $s_M = \arg\max  \left| y^*(s) - y_0^*(s) \right|$ with $s_M$ belonging to a given $\mathcal{S}_l$ where $y^*(s) \le a_{j_l+1}$ and $y_0^*(s) \le a_{k_l+1}$. 
For construction $ |a_{l_l+1} - a_{k_l+1}| \leq M_{\epsilon^*}$ and so for $M_{\epsilon^*} \to 0 $ we have $a_{l_l+1} = a_{k_l+1}$. 
Considering that $\max  \left| y^*(s) - y_0^*(s) \right| \to 0$ then $\left| y^*(s) - y_0^*(s) \right| \to 0$ for all $s \in \mathcal{S}$ leading also to  $\max \delta_l \to 0$.
Whereas the absolute value of the difference $ \left| y^*(s) - y_0^*(s) \right|$ is bounded and continuous we have that if $\int_{\mathcal{S}} \left| y^*(s) - y_0^*(s) \right| \d s $ goes to zero, also $\lim \sup_{\mathcal{S}} |y_0^*(s)-y^*(s)|$ goes to zero and hence also $M_{\epsilon^*}$.

The fact that $\zeta(\cdot;y_0)$ is non decreasing follows directly from its definition.
\end{proof}

By Lemma \ref{lem:neigh} with suitable $\epsilon^*$ we have
\[
\Pi\{ \eta_{\epsilon}(y_0) \} = \Pi[h\{ \eta_{\epsilon^*}(y^*_0) \}] =  \Pi^*\{ \eta_{\epsilon^*}(y^*_0) \} > 0.
\]
\end{proof}


\begin{proof}[of Theorem \ref{theo:consistency}]
Since $y_0(s_i)$ is equal to the observed $y_i$ for all $i$, 
we can rewrite the posterior (\ref{eq:posterior}) as 
\begin{eqnarray*}
	&& \Pi \left\{ y \in \eta_\epsilon^C(y_0) \mid y_1, \dots, y_n	\right\}  =  \\
		 &&= \frac{
			\int_{\eta_\epsilon^C(y_0) \cap \mathcal{C}_n} \prod_{i=1}^n   \delta_{y_i}(y_i)  \d \Pi(y)+
			\int_{\eta_\epsilon^C(y_0) \cap \mathcal{C}_n^C} \prod_{i=1}^n \delta_{y_i}(y_i)  \d \Pi(y) 
                  }{
			\int_{\mathcal{C}} \prod_{i=1}^n \delta_{y_i}(y_i)  \d \Pi(y)
			} \\
		&&\leq \Phi_n +  
		\frac{
			(1-\Phi_n) \int_{\eta_\epsilon^C(y_0) \cap \mathcal{C}_n} \prod_{i=1}^n   \delta_{y_i}(y_i)  \d \Pi(y)+
			\int_{\eta_\epsilon^C (y_0) \cap \mathcal{C}_n^C} \prod_{i=1}^n \delta_{y_i}(y_i)  \d \Pi(y) 
                  }{
			\int_{\mathcal{C}} \prod_{i=1}^n \delta_{y_i}(y_i)  \d \Pi(y)
			} \\
		&& =  \Phi_n +  \frac{ I_{1,n}(y_1,\dots,y_n) + I_{2,n}(y_1,\dots,y_n)}{I_{3,n}(y_1,\dots,y_n)},
\end{eqnarray*}
where $\delta_a$ is a delta mass at $a$, $\Phi_n$ is a test function and $\mathcal{C}_n$ is a sieve that grows eventually to the whole space $\mathcal C$. It suffices to show that
\begin{eqnarray}
	\Phi_n &\to& 0  \label{Itype} \\
	e^{\beta_1 n}I_{1,n}(y_1, \dots, y_n) &\to& 0  \label{IItype} \\
	e^{\beta_2 n}I_{2,n}(y_1, \dots, y_n) &\to& 0  \label{outsieve} \\
	e^{\beta n}  I_{3,n}(y_1, \dots, y_n) &\to& \infty  \label{domain}
\end{eqnarray}
with $\beta < \min\{\beta_1,\beta_2\}$.

Denote $\lfloor a \rfloor$ the integer part of $a$ and let 
$
\mathcal{S} = \bigcup_{j=1}^{\lfloor n^{1/p}\rfloor^p} \mathcal{G}_j
$
with $\mathcal{G}_j$ an $L^{\infty}$ ball of size $0.5 ( \lfloor  n^{1/p}\rfloor)^{-1} $ and center $s_j'$, where the centers are chosen on a grid so that $\lfloor n^{1/p}\rfloor^p$ balls cover $\mathcal{S}$ and each $\mathcal{G}_j$ contains at least one element of $(s_1,\ldots, s_n)^T$ under Assumption 2.
Define $X_i = 1\{y(s_i) = y_0(s_j')\}$ with $s_j'$ being the centroid of the $\mathcal{G}_j$ in which $s_i$ is contained. Let \mbox{$\Phi_n = 1\{\sum_{i=1}^n X_i < n\}$} the test on the set 
\begin{equation}
	\mathcal{C}_n = \left\{ y: y \text{ is constant in } \mathcal{G}_j, \text{for all $j= 1, \dots, \lfloor n^{1/p}\rfloor^p$}, ||y||_\infty < M_n  \right\} 
\label{eq:sieve}
\end{equation}
with $M_n = \mathcal{O}(n^\alpha)$ and $1/2 < \alpha < 1$. The first condition on the sieve governs the regularity of the process while the second gives an upper bound for the infinity norm as in \citet{choi:sche:2007}.
The true $y_0$ belongs to $\mathcal{C}_n$ for a given $n$ and hence for $n$ sufficiently large the test functions have exactly zero type~I and type~II probability.
From this (\ref{Itype}) is directly verified. 
We continue to prove (\ref{IItype}). By Fubini's theorem we have 
\begin{eqnarray*}
	E_{y_0} \{ I_{1,n}(y_1, \dots, y_n) \} 
		&=& E_{y_0} \left\{ (1-\Phi_n) \int_{\eta_\epsilon^C (y_0)\cap \mathcal{C}_n^C} \prod_{i=1}^n   \delta_{y_i}(y_i)  \d \Pi(y)  \right\} \\
		&=& \int_{\eta_\epsilon^C(y_0) \cap \mathcal{C}_n^C}  E_{y} \{ (1-\Phi_n) \} = 0
\end{eqnarray*}
where the final equality is directly verified by the test construction. 
Next we prove (\ref{outsieve}). Again by Fubini's theorem we have
\begin{eqnarray*}
	E_{y_0} \{ I_{2,n}(y_1, \dots, y_n) \} 
		&=& E_{y_0} \left\{ \int_{\eta_\epsilon^C (y_0)\cap \mathcal{C}_n^C} \prod_{i=1}^n   \delta_{y_i}(y_i)  \d \Pi(y)  \right\} \\
		&\leq& \Pi(\mathcal{C}_n^C) \\
		&\leq& c_1 e^{-c_2 n},
\end{eqnarray*}
and hence  for $\beta_2 < c_2$,
\[
	e^{\beta_2 n} I_{2,n}(y_1, \dots, y_n) \to 0 . 
\]
Finally the prior positivity of $\Pi$ makes $I_{3,n}(y_1, \dots, y_n)$ to be positive. This proves also (\ref{domain}) and concludes the proof.
\end{proof}

\bibliographystyle{biometrika} 
\bibliography{biblio}

\end{document}